\documentclass[11pt,a4paper]{article}
\usepackage[margin=25mm]{geometry}
\usepackage[T1]{fontenc}
\usepackage[utf8]{inputenc}
\usepackage{lmodern}
\usepackage{amsmath,amssymb,bm,amsthm,graphicx,booktabs,array}
\usepackage{caption,subcaption,cite,flafter,placeins}
\usepackage[colorlinks=true,linkcolor=blue,citecolor=blue,urlcolor=blue]{hyperref}
\usepackage{titlesec}
\titleformat{\section}{\large\normalfont}{\thesection}{.7em}{}
\titleformat{name=\section,numberless}{\large\normalfont}{}{0em}{}
\renewcommand{\abstractname}{Abstract}
\renewcommand{\proofname}{Proof}
\renewenvironment{abstract}{\begin{center}\normalfont\abstractname\end{center}\begin{quotation}\small\normalfont}{\end{quotation}}
\makeatletter
\renewenvironment{proof}[1][\proofname]{\par\pushQED{\qed}\normalfont\topsep6\p@\@plus6\p@\relax\trivlist\item[\hskip\labelsep\normalfont #1.]\ignorespaces}{\popQED\endtrivlist\@endpefalse}
\makeatother
\newtheoremstyle{normalplain}{.6em}{.6em}{\normalfont}{}{\normalfont}{.}{.5em}{}
\theoremstyle{normalplain}
\newtheorem{proposition}{Proposition}

\newcommand{\ket}[1]{|#1\rangle}
\newcommand{\bra}[1]{\langle#1|}
\newcommand{\proj}[1]{\ket{#1}\bra{#1}}
\newcommand{\tr}{\operatorname{tr}}
\newcommand{\Id}{\mathbb I}
\newcommand{\N}{N_{A|BCD}}
\newcommand{\CJ}{C_{\mathrm{joint}}}
\newcommand{\CL}{C_{\rightarrow}}
\newcommand{\CSEP}{C_{\mathrm{sep}}}
\newcommand{\CLOCC}{C_{\mathrm{LOCC}}}
\newcommand{\G}{\mathcal G}
\newcommand{\Diss}[1]{\mathcal D[#1]}
\newcommand{\norm}[1]{\left\|#1\right\|_1}
\newcommand{\positive}[1]{\left[#1\right]_+}
\let\oldthebibliography\thebibliography
\let\endoldthebibliography\endthebibliography
\renewenvironment{thebibliography}[1]{\oldthebibliography{#1}\setlength{\itemsep}{.15em}\setlength{\parskip}{0pt}}{\endoldthebibliography}
\allowdisplaybreaks
\title{Joint-measurement advantage and finite-time transitions in accelerated detector registers}
\author{Zhiming Huang\\[.3em]
\small School of Economics and Management, Wuyi University,\\
\small Jiangmen 529020, China\\
\small Corresponding author: \href{mailto:465609785@qq.com}{465609785@qq.com}}
\date{}
\begin{document}
\maketitle
\begin{abstract}
The Unruh response thermalizes accelerated detectors, but the entanglement remaining in a register does not by itself specify what distributed measurements can recover. We study this distinction for a four-qubit cluster state whose two assisting nodes independently undergo generalized amplitude damping, while the target nodes remain isolated and the environments are inaccessible. We derive the exact joint-assisted concurrence and the exact one-way value for arbitrary local measurements. The joint optimum equals the bipartition negativity, and computational-basis measurements attain the one-way optimum. A polar-decomposition argument, together with compactness, establishes a necessary and sufficient condition for separable measurements to attain the joint optimum and proves a strict gap whenever this condition fails. For assisting detectors at equal Unruh temperature with relaxation-rate ratio $r>1$, joint advantage exists precisely when the thermal excitation factor satisfies $q<1/r$. At finite positive temperature it ends at a unique finite time, while the remaining entanglement is nonzero. For $r=2$, both the transition time and remaining resource have closed expressions. These results identify when joint assistance is necessary within the independent weak-coupling thermalization model.
\end{abstract}

\section{Introduction}

Cluster and graph states provide multipartite entanglement for measurement-based quantum computation and quantum networks\cite{briegel2001,raussendorf2001}. Measurements on assisting nodes convert this entanglement into a resource for a target pair. Environmental noise changes the available resource, with effects that depend on the state structure and the nodes exposed to the environment\cite{cavalcanti2009,aolita2010}.

Entanglement extraction also depends on the allowed measurements. Localizable entanglement describes extraction by local measurements\cite{popp2005}, while joint assistance permits collective measurements on the helpers. Local and joint assistance have been compared for pure four-qubit states\cite{laustsen2003}. Allowing target parties to participate in two-way classical communication with the helpers can further increase the extractable entanglement\cite{gour2006}. For noisy graph states, witnesses and specified measurements give lower bounds\cite{amaro2018}; recent work establishes resource bounds and transformation criteria for local and global localization in pure multipartite states\cite{vairogs2026}.

Accelerated detector registers realize this extraction problem under thermal noise. A four-qubit cluster-state study found exact freezing of two bipartition negativities when only one detector interacts with the field\cite{li2026}. Generalized amplitude damping describes thermal relaxation by a completely positive, trace-preserving channel\cite{khatri2020}, whose excitation factor can arise from the Unruh response\cite{unruh1976}. This detector realization requires the stationary weak-coupling and master-equation validity conditions\cite{davies1974,han2025}.

The question is whether distributed measurements can fully recover the entanglement remaining after this two-node thermalization. Bipartition negativity quantifies entanglement between a target node and the rest of the register\cite{vidal2002}; concurrence quantifies the two-qubit entanglement recovered for the target pair\cite{wootters1998}. We compare these resources and relate the recovered concurrence to teleportation fidelity\cite{horodecki1999}. The inaccessible environments make the physical four-qubit mixed state the input to the assistance protocol.

We derive the exact joint-assisted concurrence and the exact one-way optimum over arbitrary local measurements. The joint value equals the bipartition negativity, and computational-basis measurements attain the one-way value. Polar decomposition and compactness establish precisely when separable measurements attain the joint optimum and prove a strict gap otherwise. The condition compares the thermal population shifts with the difference between the two relaxation factors.

For equal-temperature accelerated helpers with unequal couplings, this condition gives an exact time--temperature boundary and an acceleration threshold. At finite positive temperature, joint advantage ends at a unique finite time although entanglement remains nonzero. The transition time and remaining resource are explicit for relaxation-rate ratio two. Thus the transition changes the measurements needed to extract the remaining resource.

Section~\ref{sec:model} defines the state, channels, and protocols. Section~\ref{sec:joint} derives the joint optimum. Section~\ref{sec:local} solves the one-way optimization and proves the separable-measurement obstruction. Section~\ref{sec:thermal} gives the Unruh-parameter transition. Section~\ref{sec:conclusion} concludes.

\section{Initial state, thermalization channels, and assistance protocols}\label{sec:model}

Let $A,B,C,D$ be four qubits, with computational basis states $\ket0$ and $\ket1$ denoting the ground and excited states. The single-qubit Pauli matrices are
\begin{equation}
\sigma_x=\begin{pmatrix}0&1\\1&0\end{pmatrix},\qquad
\sigma_y=\begin{pmatrix}0&-i\\i&0\end{pmatrix},\qquad
\sigma_z=\begin{pmatrix}1&0\\0&-1\end{pmatrix}.
\label{eq:pauli-def}
\end{equation}
The initial state is
\begin{equation}
\ket{CL_4}=\frac12\bigl(\ket{0000}+\ket{0011}+\ket{1100}-\ket{1111}\bigr).
\label{eq:cluster}
\end{equation}
Define $\ket{\Phi^\pm}=(\ket{00}\pm\ket{11})/\sqrt2$. Grouping the terms in Eq.~\eqref{eq:cluster} by the computational basis of $AB$ gives
\begin{equation}
\ket{CL_4}=\frac1{\sqrt2}\bigl(\ket{00}_{AB}\ket{\Phi^+}_{CD}
+\ket{11}_{AB}\ket{\Phi^-}_{CD}\bigr).
\label{eq:cluster-split}
\end{equation}
The two environments are initially uncorrelated with the four-qubit state. During the noisy evolution, $A$ and $B$ remain isolated, whereas $C$ and $D$ undergo channels $\mathcal E_C$ and $\mathcal E_D$. Neither environment is accessible during assistance.

For $\nu=C,D$, let $\gamma_{\downarrow\nu}$ and $\gamma_{\uparrow\nu}$ denote the de-excitation and excitation rates. Define the total relaxation rate $\Gamma_\nu$ and equilibrium population imbalance $b_\nu$ for $\Gamma_\nu>0$ by
\begin{equation}
\Gamma_\nu=\gamma_{\downarrow\nu}+\gamma_{\uparrow\nu},\qquad
b_\nu=\frac{\gamma_{\downarrow\nu}-\gamma_{\uparrow\nu}}{\Gamma_\nu}.
\label{eq:rates}
\end{equation}
The equilibrium ground- and excited-state probabilities are $(1+b_\nu)/2$ and $(1-b_\nu)/2$. Positive temperatures and the zero-temperature limit give $0\le b_\nu\le1$. With $\tau$ denoting elapsed proper time, the interaction-picture thermalization generator is
\begin{equation}
\frac{d\rho_\nu}{d\tau}
=\gamma_{\downarrow\nu}\Diss{\sigma_-}\rho_\nu
+\gamma_{\uparrow\nu}\Diss{\sigma_+}\rho_\nu,
\qquad
\Diss{L}\rho=L\rho L^\dagger-\frac12\{L^\dagger L,\rho\},
\label{eq:lindblad}
\end{equation}
Here $\sigma_-=\ket0\bra1$, $\sigma_+=\sigma_-^\dagger$, and $\{M,N\}=MN+NM$. The removed local Hamiltonian phases can be compensated by local rotations and do not change the optimized resources.

Let $\tr$ denote the matrix trace. In matrix elements, Eq.~\eqref{eq:lindblad} reads
\begin{equation}
\dot\rho_{11}=-\Gamma_\nu(\rho_{11}-\frac{1-b_\nu}{2}\tr\rho_\nu),\qquad
\dot\rho_{01}=-\frac{\Gamma_\nu}{2}\rho_{01}.
\label{eq:ode-elements}
\end{equation}
Setting $u_\nu=e^{-\Gamma_\nu\tau}$ and solving these equations yields the channel
\begin{equation}
\mathcal E_\nu\!\begin{pmatrix}\rho_{00}&\rho_{01}\\\rho_{10}&\rho_{11}\end{pmatrix}
=\begin{pmatrix}
u_\nu\rho_{00}+\frac{1+b_\nu}{2}(1-u_\nu)\tr\rho&\sqrt{u_\nu}\rho_{01}\\
\sqrt{u_\nu}\rho_{10}&u_\nu\rho_{11}+\frac{1-b_\nu}{2}(1-u_\nu)\tr\rho
\end{pmatrix}.
\label{eq:channel}
\end{equation}
Retaining $\tr\rho$ makes the expression valid for unnormalized matrices and non-Hermitian operators as well as density matrices. The same channel is specified by the Kraus operators
\begin{align}
K_{0\nu}&=\sqrt{\frac{1+b_\nu}{2}}\begin{pmatrix}1&0\\0&\sqrt{u_\nu}\end{pmatrix},&
K_{1\nu}&=\sqrt{\frac{1+b_\nu}{2}}\begin{pmatrix}0&\sqrt{1-u_\nu}\\0&0\end{pmatrix},\\
K_{2\nu}&=\sqrt{\frac{1-b_\nu}{2}}\begin{pmatrix}\sqrt{u_\nu}&0\\0&1\end{pmatrix},&
K_{3\nu}&=\sqrt{\frac{1-b_\nu}{2}}\begin{pmatrix}0&0\\\sqrt{1-u_\nu}&0\end{pmatrix}.
\label{eq:kraus}
\end{align}
Let $\Id_d$ denote the $d$-dimensional identity matrix. Direct multiplication gives $\sum_jK_{j\nu}^\dagger K_{j\nu}=\Id_2$, proving complete positivity and trace preservation. The channel theorems below only require $u_\nu,b_\nu\in[0,1]$. Exponential time dependence is used for the thermalization application. An uncoupled helper has the identity channel $u_\nu=1$, for which the equilibrium parameter does not affect the map.

Write $\mathcal E=\mathcal E_C\otimes\mathcal E_D$. The four-qubit state is
\begin{equation}
\rho=(\mathrm{id}_{AB}\otimes\mathcal E)\proj{CL_4}.
\label{eq:rho-physical}
\end{equation}
Introduce three operators on $CD$:
\begin{equation}
P=\mathcal E(\proj{\Phi^+}),\qquad Q=\mathcal E(\proj{\Phi^-}),\qquad
X=\mathcal E(\ket{\Phi^+}\bra{\Phi^-}).
\label{eq:PQX}
\end{equation}
Then
\begin{equation}
\rho=\frac12\bigl(\proj{00}\otimes P+\proj{11}\otimes Q
+\ket{00}\bra{11}\otimes X+\ket{11}\bra{00}\otimes X^\dagger\bigr).
\label{eq:rho-block}
\end{equation}
Trace preservation gives $\tr P=\tr Q=1$ and $\tr X=0$. Thus the unconditional reduced state remains
\begin{equation}
\rho_{AB}=\frac12\proj{00}+\frac12\proj{11}.
\label{eq:AB-fixed}
\end{equation}
This state is separable, but it also has an equal-probability Bell-state decomposition:
\begin{equation}
\frac12\proj{\Phi^+}+\frac12\proj{\Phi^-}
=\frac12\proj{00}+\frac12\proj{11}=\rho_{AB}.
\label{eq:bell-ensemble}
\end{equation}
The weights $1/2$ are probabilities: the opposite off-diagonal terms of the two Bell projectors cancel when their weights are equal. Each Bell state has concurrence $1$, so this decomposition has average concurrence $1$. It need not be realizable by measuring $CD$ alone, because the purifying environment is inaccessible. The assisted resource therefore depends on the full accessible state, not only on $\rho_{AB}$.

Joint assistance allows an arbitrary positive operator-valued measurement (POVM) on $CD$, with effects $M_m\ge0$ and $\sum_mM_m=\Id_4$. In one-way assistance, one helper first performs an arbitrary local POVM and sends its outcome to the other helper, who performs an outcome-dependent local POVM. Either order is allowed. Both helpers then send their outcomes to $A,B$. The targets do not act during the helper measurements, no additional shared entangled ancilla is supplied, and every branch contributes to the average.

The unnormalized $AB$ state associated with effect $M_m$ is
\begin{equation}
\widetilde\rho_{AB}^{(m)}=\frac12
\begin{pmatrix}
\tr(M_mP)&0&0&\tr(M_mX)\\
0&0&0&0\\0&0&0&0\\
\tr(M_mX)^*&0&0&\tr(M_mQ)
\end{pmatrix},\qquad
p_m=\tr\widetilde\rho_{AB}^{(m)}.
\label{eq:post-state}
\end{equation}
To derive its resource value, denote any normalized two-qubit branch by $\sigma$ and write its nonzero block as $\begin{pmatrix}a&z\\z^*&d\end{pmatrix}$, where $a+d=1$ and $|z|\le\sqrt{ad}$. The Wootters formula uses the spin-flipped state $\widehat\sigma=(\sigma_y\otimes\sigma_y)\sigma^*(\sigma_y\otimes\sigma_y)$. Denote the nonzero block by $R$. Its spin flip and product are
\begin{equation}
\widehat R=\begin{pmatrix}d&z\\z^*&a\end{pmatrix},\qquad
R\widehat R=\begin{pmatrix}ad+|z|^2&2az\\2dz^*&ad+|z|^2\end{pmatrix}.
\label{eq:branch-spinflip}
\end{equation}
The diagonal entries of this product coincide, and the product of its off-diagonal entries is $4ad|z|^2$. Its eigenvalues are therefore $ad+|z|^2\pm2\sqrt{ad}|z|$. Equivalently, the eigenvalues of $\sigma\widehat\sigma$ are $(\sqrt{ad}+|z|)^2$, $(\sqrt{ad}-|z|)^2$, $0$, and $0$. Consequently,
\begin{equation}
C(\sigma)=\max\{0,(\sqrt{ad}+|z|)-(\sqrt{ad}-|z|)\}=2|z|.
\label{eq:maxcorr-concurrence}
\end{equation}
Dividing the off-diagonal element of Eq.~\eqref{eq:post-state} by $p_m$ gives
\begin{equation}
p_m C\!\left(\widetilde\rho_{AB}^{(m)}/p_m\right)
=|\tr(M_mX)|.
\label{eq:branch-resource}
\end{equation}
We denote the optimal average concurrences for joint and one-way assistance by $\CJ$ and $\CL$. A zero-probability branch contributes zero.

\section{Bipartition negativity and joint assistance}\label{sec:joint}

We normalize negativity to give value $1$ for a Bell state:
\begin{equation}
\N=\norm{\rho^{T_A}}-1,
\label{eq:negativity-def}
\end{equation}
Here $T_A$ denotes partial transposition on $A$, and $\norm{Y}=\tr\sqrt{Y^\dagger Y}$ is the trace norm. The following reduction holds for any completely positive, trace-preserving channel on $CD$, including correlated channels.

\begin{proposition}\label{prop:trace-norm}
For the state in Eq.~\eqref{eq:rho-block},
\begin{equation}
N_{A|BCD}=N_{B|ACD}=\CJ=\norm{X}.
\label{eq:trace-identity}
\end{equation}
\end{proposition}
\begin{proof}
Partial transposition on $A$ maps $\ket{00}\bra{11}$ to $\ket{10}\bra{01}$. Reordering the $AB$ support, $\rho^{T_A}$ is a direct sum of $P/2$, $Q/2$, and
\begin{equation}
\frac12\begin{pmatrix}0&X\\X^\dagger&0\end{pmatrix}.
\label{eq:pt-block}
\end{equation}
If $s_j$ are the singular values of $X$, this last block has eigenvalues $\pm s_j/2$. Since $P,Q$ are positive and each has trace $1$, $\norm{\rho^{T_A}}=1+\sum_js_j$. Partial transposition on $B$ gives the same spectrum.

Equation~\eqref{eq:branch-resource} expresses the joint optimum as $\sup_{\{M_m\}}\sum_m|\tr(M_mX)|$. Choose a unit-modulus phase $\zeta_m$ for each effect so that $\zeta_m\tr(M_mX)$ is nonnegative and real, and set $K=\sum_m\zeta_mM_m$. For arbitrary unit vectors $v,w$, the Cauchy--Schwarz inequality gives
\begin{equation}
|\bra vK\ket w|
\le\sum_m\|M_m^{1/2}v\|\,\|M_m^{1/2}w\|
\le\left(\sum_m\bra vM_m\ket v\right)^{1/2}
\left(\sum_m\bra wM_m\ket w\right)^{1/2}=1.
\label{eq:K-contraction}
\end{equation}
Thus $\|K\|_\infty\le1$, where $\|K\|_\infty$ is the largest singular value of $K$. Trace-norm duality yields
\begin{equation}
\sum_m|\tr(M_mX)|=\mathrm{Re}\tr(KX)\le\norm{X}.
\label{eq:joint-upper}
\end{equation}
Take the polar decomposition $X=V|X|$, where $|X|=\sqrt{X^\dagger X}$, extending $V$ to a unitary on its kernel when necessary. Write the spectral decomposition $V^\dagger=\sum_j\zeta_j\Pi_j$, with $|\zeta_j|=1$. Since $\Pi_jV=\zeta_j^*\Pi_j$, measurement in these spectral projectors gives
\begin{equation}
\begin{aligned}
\tr(\Pi_jX)&=\zeta_j^*\tr(\Pi_j|X|),\qquad \tr(\Pi_j|X|)\ge0,\\
\sum_j|\tr(\Pi_jX)|&=\sum_j\tr(\Pi_j|X|)=\tr|X|=\norm X.
\end{aligned}
\label{eq:joint-attain}
\end{equation}
Together with the upper bound, this proves equality.
\end{proof}

To evaluate the resources for two thermalized nodes, define
\begin{equation}
w=u_Cb_D(1-u_D)+u_Db_C(1-u_C),\qquad
\delta=|u_C-u_D|.
\label{eq:parameters}
\end{equation}
The closed assistance formulas depend on $u_C,u_D$ and the combination $w$, which collects the thermal population shifts $b_\nu(1-u_\nu)$ weighted by the other helper's relaxation factor. The parameter $\delta$ measures the difference between the two relaxation factors. The parameter range gives
\begin{equation}
\begin{aligned}
0\le w&\le u_C+u_D-2u_Cu_D\le\min\{u_C+u_D,1\},\\
1-(u_C+u_D-2u_Cu_D)&=(1-u_C)(1-u_D)+u_Cu_D\ge0.
\end{aligned}
\label{eq:displacement-bound}
\end{equation}
The Pauli operators are denoted by $\sigma_{x\nu},\sigma_{y\nu},\sigma_{z\nu}$, while the unsubscripted $X$ always denotes the four-dimensional operator in Eq.~\eqref{eq:PQX}.

The channel acts on the operator basis as
\begin{equation}
\mathcal E_\nu(\Id)=\Id+b_\nu(1-u_\nu) \sigma_{z\nu},\qquad
\mathcal E_\nu(\sigma_{z\nu})=u_\nu \sigma_{z\nu},\qquad
\mathcal E_\nu(\sigma_{x\nu})=\sqrt{u_\nu}\sigma_{x\nu},\qquad
\mathcal E_\nu(\sigma_{y\nu})=\sqrt{u_\nu}\sigma_{y\nu}.
\label{eq:pauli-action}
\end{equation}
The initial cross operator is
\begin{equation}
\begin{aligned}
\ket{\Phi^+}\bra{\Phi^-}
&=\frac12\bigl(\proj{00}-\ket{00}\bra{11}+\ket{11}\bra{00}-\proj{11}\bigr)\\
&=\frac14\bigl(\sigma_{zC}+\sigma_{zD}-i\sigma_{xC}\sigma_{yD}-i\sigma_{yC}\sigma_{xD}\bigr).
\end{aligned}
\label{eq:X-initial-pauli}
\end{equation}
Identity tensor factors are omitted: for example, $\sigma_{zC}$ means $\sigma_{zC}\otimes\Id_D$. Applying Eq.~\eqref{eq:pauli-action} term by term gives
\begin{equation}
X=\frac14\bigl[u_C\sigma_{zC}+u_D\sigma_{zD}+w\sigma_{zC}\sigma_{zD}
-i\sqrt{u_Cu_D}(\sigma_{xC}\sigma_{yD}+\sigma_{yC}\sigma_{xD})\bigr].
\label{eq:X-pauli}
\end{equation}
In the ordered basis $\ket{00},\ket{01},\ket{10},\ket{11}$,
\begin{equation}
X=\frac14\begin{pmatrix}
u_C+u_D+w&0&0&-2\sqrt{u_Cu_D}\\
0&u_C-u_D-w&0&0\\
0&0&-u_C+u_D-w&0\\
2\sqrt{u_Cu_D}&0&0&w-u_C-u_D
\end{pmatrix}.
\label{eq:X-matrix}
\end{equation}

\begin{proposition}\label{prop:joint-closed}
After the two independent thermalization channels, the bipartition negativity and optimal joint-assisted concurrence are
\begin{equation}
\N=\CJ=\frac12\left[\max\{w,\delta\}
+\max\left\{u_C+u_D,\sqrt{w^2+4u_Cu_D}\right\}\right].
\label{eq:N-closed}
\end{equation}
\end{proposition}
\begin{proof}
Here even parity refers to the span of $\ket{00},\ket{11}$ and odd parity to the span of $\ket{01},\ket{10}$, according to whether the number of excited qubits is even or odd. The two singular values in the odd-parity part of Eq.~\eqref{eq:X-matrix} sum to
\begin{equation}
\frac{|u_C-u_D-w|+|-u_C+u_D-w|}{4}
=\frac12\max\{w,\delta\}.
\label{eq:odd-norm}
\end{equation}
The even-parity block is
\begin{equation}
B=\frac14\begin{pmatrix}u_C+u_D+w&-2\sqrt{u_Cu_D}\\2\sqrt{u_Cu_D}&w-u_C-u_D\end{pmatrix},\qquad
\det B=\frac{w^2-\delta^2}{16}.
\label{eq:even-block}
\end{equation}
For a $2\times2$ matrix with singular values $s_1,s_2$,
\begin{equation}
(s_1+s_2)^2=\tr(B^\dagger B)+2|\det B|.
\label{eq:two-singular}
\end{equation}
Here $s_1s_2=|\det B|$, since $\det(B^\dagger B)=|\det B|^2$. Also, $\tr(B^\dagger B)=((u_C+u_D)^2+w^2+4u_Cu_D)/8$ and $(u_C+u_D)^2-4u_Cu_D=\delta^2$. Substitution therefore yields
\begin{equation}
\norm B^2=\frac{(u_C+u_D)^2+w^2+4u_Cu_D+|w^2-\delta^2|}{8}
=\begin{cases}
(u_C+u_D)^2/4,&w\le\delta,\\
(w^2+4u_Cu_D)/4,&w>\delta.
\end{cases}
\label{eq:even-norm-squared}
\end{equation}
Taking the nonnegative square root gives
\begin{equation}
\norm B=
\begin{cases}
(u_C+u_D)/2,& w\le\delta,\\[.2em]
\sqrt{w^2+4u_Cu_D}/2,&w>\delta.
\end{cases}
\label{eq:even-norm}
\end{equation}
Adding the two parity contributions and applying Proposition~\ref{prop:trace-norm} proves the result.
\end{proof}

When $w\le\delta$, the joint optimum is $\max(u_C,u_D)$ and is attained by computational-basis measurements. When $w>\delta$, an optimal joint measurement basis is
\begin{equation}
\begin{aligned}
\ket{\chi_\pm}&=\frac{\ket{00}\pm i\ket{11}}{\sqrt2},\qquad
\ket{01},\qquad\ket{10},\\
\bra{\chi_\pm}X\ket{\chi_\pm}&=\frac{w\mp2i\sqrt{u_Cu_D}}4.
\end{aligned}
\label{eq:joint-basis}
\end{equation}
The absolute values in the second line sum to $\sqrt{w^2+4u_Cu_D}/2$. Since both odd-parity diagonal entries are negative when $w>\delta$, their absolute values sum to $w/2$. This measurement therefore attains Eq.~\eqref{eq:N-closed}. An optimal joint measurement therefore switches between two fixed bases.

\section{Exact one-way optimum and the separable-measurement obstruction}\label{sec:local}

We optimize the one-way protocol over arbitrary first-helper POVMs and conditional local POVMs of the second helper.

\begin{proposition}\label{prop:local-closed}
For the state in Eq.~\eqref{eq:rho-physical}, either order of one-way assistance has optimum
\begin{equation}
\CL=\frac12\bigl(u_C+u_D+\max\{w,\delta\}\bigr).
\label{eq:L-closed}
\end{equation}
Computational-basis measurements by both helpers attain this value.
\end{proposition}
\begin{proof}
Suppose $C$ measures first. Refining its effects into rank-one effects cannot decrease the average concurrence, by the triangle inequality for each original cross term. It therefore suffices to consider
\begin{equation}
M_m=e_m\proj{n_m},\qquad
\proj{n_m}=\frac12(\Id+\bm n_m\cdot\bm\sigma),\qquad
\sum_me_m=2,\qquad\sum_me_m\bm n_m=0,
\label{eq:rank-one-povm}
\end{equation}
Here $e_m\ge0$ and $|\bm n_m|=1$. Writing $z=n_{mz}$ and contracting Eq.~\eqref{eq:X-pauli} on $C$ gives the conditional cross operator on $D$:
\begin{equation}
X_D^{(m)}=\frac{e_m}{4}\left[u_Cz\Id+(u_D+wz)\sigma_{zD}
-i\sqrt{u_Cu_D}(n_{mx}\sigma_{yD}+n_{my}\sigma_{xD})\right].
\label{eq:conditional-X}
\end{equation}
The measurement-optimization argument of Proposition~\ref{prop:trace-norm} also applies to a single assisting qubit. Thus the optimal contribution of $D$ in branch $m$ is $\norm{X_D^{(m)}}$. The matrix in brackets is
\begin{equation}
T=\begin{pmatrix}
u_C z+u_D+wz&-\sqrt{u_Cu_D}(n_{mx}+in_{my})\\
\sqrt{u_Cu_D}(n_{mx}-in_{my})&u_C z-u_D-wz
\end{pmatrix}.
\label{eq:conditional-T}
\end{equation}
Using $n_{mx}^2+n_{my}^2=1-z^2$, its determinant and squared Frobenius norm are
\begin{align}
\det T&=u_Cu_D+u_C(u_C-u_D)z^2-(u_D+wz)^2,\\
\tr(T^\dagger T)&=2\left[u_Cu_D+u_C(u_C-u_D)z^2+(u_D+wz)^2\right].
\label{eq:conditional-singular}
\end{align}
Equation~\eqref{eq:two-singular} then gives
\begin{equation}
\norm{X_D^{(m)}}=\frac{e_m}{2}\,g(z),\qquad
g(z)=\max\left\{|u_D+wz|,
\sqrt{u_Cu_D+u_C(u_C-u_D)z^2}\right\}.
\label{eq:g-function}
\end{equation}
The branch value depends only on $z\in[-1,1]$. Completeness enforces $\sum_me_mz_m=0$, so a bound by the chord joining $z=\pm1$ will bound every first-helper POVM.

If $u_C\ge u_D$, the square-root term is the Euclidean norm of a vector with components $\sqrt{u_Cu_D}$ and $\sqrt{u_C(u_C-u_D)}z$. These components depend affinely on $z$, so the norm is convex. The absolute-value term is also convex, and the maximum of two convex functions is convex. Hence $g$ lies below its endpoint chord:
\begin{equation}
g(z)\le\frac{1+z}{2}g(1)+\frac{1-z}{2}g(-1).
\label{eq:endpoint-chord}
\end{equation}
Using both completeness conditions in Eq.~\eqref{eq:rank-one-povm} gives
\begin{equation}
\begin{aligned}
\sum_m\frac{e_m}{2}g(z_m)
&\le\frac14\left[g(1)\sum_me_m(1+z_m)
+g(-1)\sum_me_m(1-z_m)\right]\\
&=\frac{g(1)+g(-1)}2.
\end{aligned}
\label{eq:weighted-chord}
\end{equation}

If $u_C<u_D$, then $\delta=u_D-u_C>0$. Define $h(z)=\sqrt{u_Cu_D-u_C \delta z^2}$. Since $u_D+\delta z\ge u_C\ge0$, the identity
\begin{equation}
(u_D+\delta z)^2-h(z)^2=u_D \delta(1+z)^2\ge0
\label{eq:concave-square}
\end{equation}
implies
\begin{equation}
h(z)\le u_D+\delta z.
\label{eq:concave-tangent}
\end{equation}
If $w\le \delta$, then $u_D+wz\ge u_D+\delta z\ge h(z)$ for $z\le0$, whereas $u_D+wz\ge u_D\ge h(z)$ for $z\ge0$. Thus $g(z)=u_D+wz$ and its weighted mean is $u_D$. If $w>\delta$, the bound $w\le u_C+u_D$ gives $g(-1)=u_C$ and $g(1)=u_D+w$. Their chord is
\begin{equation}
\ell(z)=\frac{u_C+u_D+w}{2}+\frac{\delta+w}{2}z,
\qquad
\ell(z)-[u_D+\delta z]=\frac{w-\delta}{2}(1+z)\ge0.
\label{eq:concave-chord-bound}
\end{equation}
The function $|u_D+wz|$ is convex and lies below its own endpoint chord. At $z=-1$, $|u_D-w|\le u_C=g(-1)$ follows from $w>\delta$ and $w\le u_C+u_D$; at $z=1$ its value is $g(1)$. Its chord is therefore no higher than $\ell(z)$. Both terms defining $g$ lie below $\ell$, so Eqs.~\eqref{eq:endpoint-chord} and~\eqref{eq:weighted-chord} apply in this case as well.

The endpoint mean can be evaluated explicitly in three ranges:
\begin{equation}
\frac{g(1)+g(-1)}2=
\begin{cases}
u_C,&u_C\ge u_D,\quad w\le u_C-u_D,\\
u_D,&u_D>u_C,\quad w\le u_D-u_C,\\
(u_C+u_D+w)/2,&w>|u_C-u_D|.
\end{cases}
\label{eq:endpoint-average}
\end{equation}
Consequently, every first-helper POVM satisfies
\begin{equation}
\sum_m\norm{X_D^{(m)}}\le\frac12\left[
\max\{u_D+w,u_C\}+\max\{|u_D-w|,u_C\}\right]
=\frac12\bigl(u_C+u_D+\max\{w,\delta\}\bigr).
\label{eq:one-way-upper}
\end{equation}
Taking the computational-basis projectors of $C$, with $z=\pm1$, makes Eq.~\eqref{eq:conditional-X} diagonal. A computational-basis measurement on $D$ attains its trace norm in each branch, so the upper bound is achievable. Exchanging $C,D$ leaves the result invariant and proves the statement for the other order. The same pointwise inequalities can be integrated for a continuous-outcome POVM, so allowing continuous outcomes leaves the optimum unchanged.
\end{proof}

Define the joint-minus-one-way gap by $\G=\CJ-\CL$. Propositions~\ref{prop:joint-closed} and~\ref{prop:local-closed} give
\begin{equation}
\G=\frac12\positive{\sqrt{w^2+4u_Cu_D}-u_C-u_D},\qquad
\G>0\ \Longleftrightarrow\ w>\delta,
\label{eq:gap}
\end{equation}
The equivalence follows from $(u_C+u_D)^2-4u_Cu_D=\delta^2$: the square root exceeds $u_C+u_D$ exactly when $w^2>\delta^2$, and $w,\delta$ are nonnegative. Here $[x]_+=\max(x,0)$. 

Let $\CLOCC$ denote the optimal assistance by local operations and classical communication between $C,D$, with any number of rounds, and let $\CSEP$ denote the optimum over all separable POVMs. Operational optima are understood as suprema when attainment has not been established. Each separable effect is a sum of positive product operators, and this class contains all such local protocols.

\begin{proposition}\label{prop:separable-barrier}
If $w\le\delta$, then
\begin{equation}
\CL=\CLOCC=\CSEP=\CJ=\max(u_C,u_D).
\label{eq:all-equal}
\end{equation}
If $w>\delta$, then
\begin{equation}
\CL\le\CLOCC\le\CSEP<\CJ=\N.
\label{eq:separable-strict}
\end{equation}
\end{proposition}
\begin{proof}
In the first case, computational-basis measurements attain the joint upper bound. In the second, both odd-parity diagonal entries of Eq.~\eqref{eq:X-matrix} are strictly negative, and the even-parity block has positive determinant. Thus $X$ is full rank. Moreover, $u_Cu_D=0$ would imply $w\le\delta$, so $\sqrt{u_Cu_D}>0$ and $w>0$ in the present case. The unitary factor in its polar decomposition $X=V|X|$ is
\begin{equation}
V\big|_{\{00,11\}}=
\frac1{\sqrt{w^2+4u_Cu_D}}\begin{pmatrix}w&-2\sqrt{u_Cu_D}\\2\sqrt{u_Cu_D}&w\end{pmatrix},\qquad
V\big|_{\{01,10\}}=-\Id_2.
\label{eq:polar-factor}
\end{equation}
The even-parity block of Eq.~\eqref{eq:polar-factor} is unitary, and that of $V^\dagger X$ is real symmetric, with
\begin{equation}
\tr\bigl[(V^\dagger X)_{\mathrm{even}}\bigr]=\frac{\sqrt{w^2+4u_Cu_D}}2>0,\qquad
\det\bigl[(V^\dagger X)_{\mathrm{even}}\bigr]=\frac{w^2-\delta^2}{16}>0.
\label{eq:polar-positivity}
\end{equation}
Both eigenvalues of this block are therefore positive. The odd-parity block is also positive, so $V^\dagger X$ is positive definite, verifying the polar decomposition. For $\|K\|_\infty\le1$, equality $\mathrm{Re}\tr(KX)=\tr|X|$ is possible only for $K=V^\dagger$. Indeed, write $|X|=\sum_j\mu_j\proj{v_j}$, with every $\mu_j>0$. Contractivity gives $\mathrm{Re}\bra{v_j}KV\ket{v_j}\le1$. Equality in the weighted sum forces each term to equal $1$. Equality in $|\bra{v_j}KV\ket{v_j}|\le\|KVv_j\|\le1$ then requires $KVv_j=v_j$. As these vectors form a basis, $KV=\Id$ and $K=V^\dagger$.

If a separable POVM attained the joint optimum, its contraction $K=\sum_m\zeta_mM_m$ in Eq.~\eqref{eq:K-contraction} would therefore equal $V^\dagger$. For a unit eigenvector $v$ of $K$ with eigenvalue $\lambda$,
\begin{equation}
\lambda=\sum_m\zeta_m\bra vM_m\ket v,\qquad
\sum_m\bra vM_m\ket v=1.
\label{eq:phase-convexity}
\end{equation}
The left-hand side has modulus $1$, so every phase with nonzero weight must equal $\lambda$. Apply this conclusion to every eigenvector. If $\bra vM_m\ket v=0$, positivity gives $\|M_m^{1/2}v\|^2=0$ and hence $M_mv=0$. Thus each effect has zero support on the eigenspaces with other phases. It can be supported only on the eigenspace of $V^\dagger$ associated with its own phase.

By Eq.~\eqref{eq:polar-factor}, the even-parity eigenspaces are the one-dimensional spaces spanned by $\ket{\chi_+}$ and $\ket{\chi_-}$. Their eigenvalues are $(w\pm2i\sqrt{u_Cu_D})/\sqrt{w^2+4u_Cu_D}$. They are distinct and neither equals the odd-parity eigenvalue $-1$, since $u_Cu_D>0$ and $w>0$. A nonzero positive effect supported on either space must be a positive multiple of $\proj{\chi_\pm}$, which is entangled and cannot be a separable effect. The vectors $\ket{\chi_\pm}$ have two nonzero Schmidt coefficients, both $1/\sqrt2$, so neither one-dimensional support contains a product vector. A nonzero separable positive effect must contain a positive product contribution supported within its range, which is impossible on either space. Completeness requires nonzero effects on both even-parity eigenspaces. This contradicts POVM completeness. Hence no separable measurement attains the joint optimum. Appendix~\ref{app:compact} proves existence of the separable optimum, converting nonattainment into a strict gap between optimal values.
\end{proof}

These resource formulas also recover the single-node freezing boundary. For $u_C=1$, $w=b_D(1-u_D)\le1-u_D=\delta$, so $\N=\CJ=\CL=1$, recovering exact freezing. If both nodes relax nontrivially, $u_C,u_D<1$, then $\N<1$. For $w>\delta$, Eq.~\eqref{eq:displacement-bound} ensures $2-w>0$, allowing the inequality to be squared without changing its direction:
\begin{equation}
\N<1\ \Longleftrightarrow\ \sqrt{w^2+4u_Cu_D}<2-w
\ \Longleftrightarrow\ w+u_Cu_D<1.
\label{eq:no-freezing-equivalence}
\end{equation}
Moreover,
\begin{equation}
w+u_Cu_D\le u_C+u_D-u_Cu_D
=1-(1-u_C)(1-u_D)<1
\label{eq:no-perfect-freezing}
\end{equation}
proves this condition. For $w\le\delta$, it follows from $\N=\max(u_C,u_D)<1$. Exact freezing and the operational gap therefore have distinct and explicit parameter boundaries.

For a communication task, each normalized branch $\sigma$ in Eq.~\eqref{eq:post-state} has Bell-state overlap $(1+C(\sigma))/2$ after phase compensation, by Eq.~\eqref{eq:maxcorr-concurrence}. Standard deterministic teleportation of a uniformly distributed unknown pure qubit then has fidelity $(2+C(\sigma))/3$. Averaging over all branches and optimizing the helper measurements gives
\begin{equation}
F_{\mathrm{joint}}=\frac{2+\CJ}{3},\qquad
F_{\rightarrow}=\frac{2+\CL}{3},\qquad
F_{\mathrm{joint}}-F_{\rightarrow}=\frac{\G}{3}.
\label{eq:fidelity}
\end{equation}
These values remain optimal when deterministic local preprocessing by $A,B$ is allowed after the helper measurement. For any two-qubit state $\sigma$ and Bell projector $\Pi_\Phi$, the H\"older inequality gives
\begin{equation}
\tr(\sigma\Pi_\Phi)=\tr(\sigma^{T_A}\Pi_\Phi^{T_A})
\le\|\sigma^{T_A}\|_1\|\Pi_\Phi^{T_A}\|_\infty
=\frac{1+N_{AB}(\sigma)}2.
\label{eq:bell-bound}
\end{equation}
Here $\|\Pi_\Phi^{T_A}\|_\infty=1/2$. Negativity cannot increase on average under local operations and classical communication, so the branches of deterministic preprocessing cannot increase this average Bell-overlap bound. Each maximally correlated branch satisfies $N_{AB}=C$ and reaches the bound by phase compensation. 

\section{Exact transitions under thermalization and Unruh parameters}\label{sec:thermal}

To fix the thermal notation, define the detailed-balance factor
\begin{equation}
q=\frac{\gamma_\uparrow}{\gamma_\downarrow},\qquad
b=\frac{1-q}{1+q},\qquad 0\le q\le1.
\label{eq:thermal-factor}
\end{equation}
The value $q=0$ is the zero-temperature limit, while $q=1$ denotes the infinite-temperature limit of a symmetric channel. First consider equal assisting-node parameters $u=e^{-\Gamma\tau}$ and $b$. Then $\delta=0$ and $w=2bu(1-u)$. The closed formulas reduce to
\begin{equation}
\begin{aligned}
\CJ=\N&=u\left[\sqrt{1+b^2(1-u)^2}+b(1-u)\right],\\
\CL&=u\left[1+b(1-u)\right],\\
\G&=u\left[\sqrt{1+b^2(1-u)^2}-1\right].
\end{aligned}
\label{eq:symmetric}
\end{equation}
At $u=1$, the two assistance classes coincide. Every $b>0$ and $0<u<1$ yields strictly positive joint advantage. For $b=0$, both resources equal $u$. 

A relativistic detector model realizes these parameters. In natural units $c=\hbar=k_B=1$, take two Unruh--DeWitt detectors with equal energy gap $\Omega>0$, each coupled to its own independent massless scalar field in four-dimensional Minkowski vacuum. The detectors have the same uniform proper acceleration $a$ and coupling strengths $\epsilon_C,\epsilon_D$. The independent fields give the product dynamics in Eq.~\eqref{eq:rho-physical}.

The interaction-picture coupling of each assisting detector is
\begin{equation}
H_{I\nu}(\tau)=\epsilon_\nu\chi_\nu(\tau)
\bigl(\sigma_{+\nu}e^{i\Omega\tau}+\sigma_{-\nu}e^{-i\Omega\tau}\bigr)
\Phi_\nu[x_\nu(\tau)],\qquad \nu=C,D.
\label{eq:detector-interaction}
\end{equation}
Here $\epsilon_\nu$ is the coupling strength, $\chi_\nu$ is a real switching function, $\sigma_{+\nu}=\ket1_\nu\bra0$ and $\sigma_{-\nu}=\ket0_\nu\bra1$ are the raising and lowering operators, and $\Phi_\nu[x_\nu(\tau)]$ is the scalar-field operator sampled along the detector trajectory, with spatial regulation understood in the microscopic model. The field states have zero one-point functions. In the stationary weak-coupling limit, each field therefore produces its own thermalization generator. 

Along the trajectory $t(\tau)=a^{-1}\sinh(a\tau)$, $x(\tau)=a^{-1}\cosh(a\tau)$, let $\xi$ denote the proper-time difference. The pulled-back Wightman function and excitation factor are
\begin{equation}
W(\xi)=-\frac{a^2}{16\pi^2\sinh^2[a(\xi-i0)/2]},\qquad
q=e^{-2\pi\Omega/a}.
\label{eq:unruh-W}
\end{equation}
For constant coupling, $\chi_\nu=1$, the stationary weak-coupling and rotating-wave limits give
\begin{equation}
\begin{aligned}
\gamma_{\downarrow\nu}&=\frac{\epsilon_\nu^2\Omega}{2\pi(1-q)},&
\gamma_{\uparrow\nu}&=\frac{q\epsilon_\nu^2\Omega}{2\pi(1-q)},\\
\Gamma_\nu&=\frac{\epsilon_\nu^2\Omega}{2\pi}\frac{1+q}{1-q}.
\end{aligned}
\label{eq:unruh-rates}
\end{equation}
Appendix~\ref{app:unruh} derives the Fourier integral by residues. These rates imply
\begin{equation}
b_C=b_D=b=\frac{1-q}{1+q},\qquad
\frac{1-b}{2}=\frac{q}{1+q}.
\label{eq:unruh-b}
\end{equation}
Equation~\eqref{eq:unruh-rates} is the stationary pointlike response. Finite detector profiles and switching require their corresponding response rates and the weak-coupling reliability conditions\cite{davies1974,han2025}. The channel results are exact for the thermalization model; these microscopic conditions govern its detector realization.

Let $r=\Gamma_D/\Gamma_C>1$ be the relaxation-rate ratio and $\theta=\Gamma_C\tau$ the dimensionless time. Equation~\eqref{eq:unruh-rates} gives $r=\epsilon_D^2/\epsilon_C^2$ for equal accelerations and gaps. Write
\begin{equation}
u_C=e^{-\theta},\qquad u_D=e^{-r\theta},\qquad
w=b(u_C+u_D-2u_Cu_D),\qquad
\delta=u_C-u_D.
\label{eq:unequal-rates}
\end{equation}
The transition below holds for any two thermalization channels with the same equilibrium factor. The Unruh temperature is one realization.

\begin{proposition}\label{prop:transition}
Let $r>1$ and $0<q<1$. A time interval with $\G>0$ exists if and only if $q<1/r$. In that case there is a unique $\theta_\times>0$ satisfying
\begin{equation}
q e^{r\theta_\times}-e^{\theta_\times}+1-q=0,
\label{eq:cross-equation}
\end{equation}
and
\begin{equation}
\G>0\quad(0<\theta<\theta_\times),\qquad
\CJ=\CL=\CLOCC=e^{-\theta}\quad(\theta\ge\theta_\times).
\label{eq:cross-regions}
\end{equation}
Equivalently, the time--temperature boundary is
\begin{equation}
q_c(\theta;r)=\frac{e^\theta-1}{e^{r\theta}-1},\qquad
\G>0\ \Longleftrightarrow\ 0\le q<q_c(\theta;r),\qquad\theta>0.
\label{eq:qc-curve}
\end{equation}
\end{proposition}
\begin{proof}
Substituting Eq.~\eqref{eq:unequal-rates} into the condition in Eq.~\eqref{eq:gap} yields
\begin{equation}
w-\delta=-\frac{2u_Cu_D}{1+q}
\left(qe^{r\theta}-e^\theta+1-q\right).
\label{eq:threshold-reduction}
\end{equation}
Set $x=e^\theta>1$ and $f(x)=qx^r-x+1-q$. Then
\begin{equation}
f(1)=0,\qquad f'(1)=qr-1,\qquad
f''(x)=qr(r-1)x^{r-2}>0.
\label{eq:convex-root}
\end{equation}
If $qr\ge1$, strict convexity gives $f(x)>0$ for all $x>1$, so there is no joint advantage. If $qr<1$, the function initially decreases to the right of $1$ and eventually tends to positive infinity because $qx^r$ dominates. Its derivative is strictly increasing and eventually positive, so it has exactly one minimum. Beyond this minimum, $f$ increases strictly to infinity and has exactly one additional zero $x_\times>1$. It is negative for $1<x<x_\times$ and positive for $x>x_\times$. Equation~\eqref{eq:threshold-reduction} gives the stated interval and uniqueness. Solving $f(x)<0$ for $q$ gives Eq.~\eqref{eq:qc-curve}. Proposition~\ref{prop:separable-barrier} establishes equality of all operation classes after the transition.
\end{proof}

At zero temperature, $q=0$, one has $f(x)=1-x<0$, so joint advantage persists for every finite positive time and the termination time diverges. For $r=1$ and $q<1$, Eq.~\eqref{eq:symmetric} likewise gives no finite termination time. For an Unruh realization with $r>1$, the acceleration condition for an advantage interval is
\begin{equation}
a<\frac{2\pi\Omega}{\ln r}.
\label{eq:acceleration-threshold}
\end{equation}

For $r=2$, Eq.~\eqref{eq:cross-equation} is the quadratic equation
\begin{equation}
qx^2-x+1-q=(x-1)(qx+q-1)=0.
\label{eq:ratio-two-polynomial}
\end{equation}
Excluding the initial root $x=1$, the positive finite transition time is
\begin{equation}
\theta_\times=\ln\frac{1-q}{q},\qquad
\tau_\times=\frac1{\Gamma_C}\ln\frac{1-q}{q},\qquad
0<q<\frac12.
\label{eq:ratio-two-time}
\end{equation}
At this time,
\begin{equation}
\CJ(\tau_\times)=\CL(\tau_\times)=u_C(\tau_\times)=\frac{q}{1-q}>0.
\label{eq:resource-at-transition}
\end{equation}
The transition changes the operations required to extract the remaining resource: before the transition, joint measurements strictly outperform all separable measurements; afterwards, computational-basis measurements extract the full bipartition resource.

Fig.~\ref{fig:transition}(a) shows the exact boundary for $r=2$. Here the horizontal coordinate is $\theta=\Gamma_C\tau$, using the relaxation-time unit at each temperature. Changing $a$ changes both $q$ and $\Gamma_C$; a comparison at fixed physical time must also use Eq.~\eqref{eq:unruh-rates}. The color is the joint-minus-one-way gap $\G$; its positive region also identifies where all separable measurements fall below the joint optimum. For $\theta>0$, advantage occurs below $q_c=(1+e^\theta)^{-1}$. This curve decreases from the limiting value $1/2$ to zero, so increasing $q$ shortens the advantage interval. The gap is zero at $\theta=0$ for every $q$. Physically, the boundary is $w=\delta$: the thermalization parameter combination ceases to exceed the relaxation asymmetry. The dashed line $q=0.2$ selects panel (b).

In Fig.~\ref{fig:transition}(b), the solid and dashed resource curves are $\CJ=\N$ and $\CL$, while the dotted curve is $u_C=e^{-\theta}$. For $0<\theta<\ln4$, the joint value exceeds the one-way value, and Proposition~\ref{prop:separable-barrier} excludes attainment of the joint curve by any separable measurement. The shading marks this interval; the vertical distance between the two resource curves is $\G$. At $\theta_\times=\ln4$, the three curves meet at $1/4$ and coincide thereafter. The nonzero crossing value shows that joint advantage can end while entanglement survives; after the crossing, computational-basis measurements extract the full remaining resource.

\begin{figure}[htbp]
\centering
\begin{subfigure}[t]{.49\linewidth}
\includegraphics[width=\linewidth]{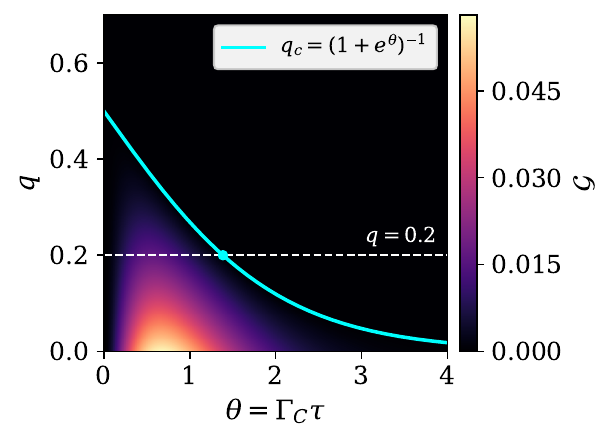}
\caption{Exact time--temperature boundary.}
\end{subfigure}\hfill
\begin{subfigure}[t]{.49\linewidth}
\includegraphics[width=\linewidth]{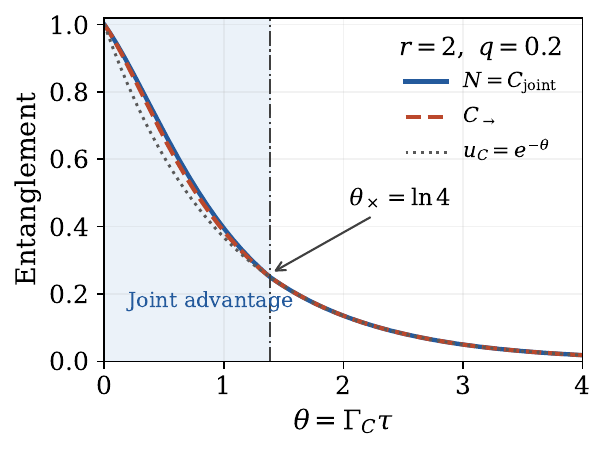}
\caption{Finite-time transition for $q=0.2$.}
\end{subfigure}
\caption{Operational transition for relaxation-rate ratio $r=2$. The left panel shows the concurrence gap $\G$ from Eq.~\eqref{eq:gap}. The boundary $q_c=(1+e^\theta)^{-1}$ follows directly from Eq.~\eqref{eq:qc-curve}; the gap is exactly zero above it. The right panel uses Eqs.~\eqref{eq:N-closed} and~\eqref{eq:L-closed}. The shaded region is where all separable measurements fall below the joint optimum. At $\theta_\times=\ln4$, both resource curves meet $u_C=e^{-\theta}$ at $1/4$ and coincide thereafter.}\label{fig:transition}
\end{figure}
\FloatBarrier

\section{Discussion and conclusions}\label{sec:conclusion}

For a four-qubit cluster state with independently thermalized helpers, joint-assisted concurrence equals the bipartition negativity, while computational-basis measurements attain the exact one-way optimum. The optimum over separable measurements equals the joint optimum precisely when $w\le|u_C-u_D|$; outside this region their optimal value is strictly smaller.

For equal-temperature accelerated detectors with relaxation-rate ratio $r>1$, an advantage interval exists exactly when $q<1/r$. At finite positive temperature it ends at a unique finite time, with nonzero entanglement remaining. The explicit $r=2$ solution distinguishes the loss of joint-measurement advantage from the loss of the resource itself. These detector conclusions hold within the independent stationary weak-coupling model.

The precise multiround LOCC optimum in the advantage region remains open; it lies between the one-way value and the separable optimum, strictly below joint assistance.

\appendix
\section{Existence of the separable-assisted optimum}\label{app:compact}

This appendix supplies the compactness argument required for the strict gap in Proposition~\ref{prop:separable-barrier}. Separable positive effects can be refined into nonnegative combinations of normalized product projectors. Refinement cannot decrease $\sum_m|\tr(M_mX)|$. Hence it suffices to optimize over
\begin{equation}
M_m=t_m\proj{c_m d_m},\qquad t_m\ge0,\qquad
\sum_m M_m=\Id_4,\qquad\sum_mt_m=4.
\label{eq:sep-refine}
\end{equation}
Associate a unit-modulus phase $\zeta_m$ with each effect. Maximizing the objective is equivalent to maximizing $\mathrm{Re}\tr(KX)$, where
\begin{equation}
(\Id_4,K)=\sum_mt_m\bigl(\proj{c_md_m},\zeta_m\proj{c_md_m}\bigr).
\label{eq:convex-pair}
\end{equation}
Normalized product projectors and unit-circle phases form a compact set. Dividing the right-hand side of Eq.~\eqref{eq:convex-pair} by $4$ puts it in the convex hull of this set. The first component is a $4\times4$ Hermitian matrix of real dimension $16$; the second is a complex matrix of real dimension $32$. The total real dimension is at most $48$. By Carath\'eodory's theorem, every point has a representation with at most $49$ terms. The convex hull is therefore the continuous image of a product of $49$ copies of the compact set and a finite-dimensional closed simplex, and is compact. Fixing the first component to $\Id_4/4$ is a closed constraint, so the allowed set of second components $K$ is compact. No infinite-outcome measurement is needed to attain the maximum.

The continuous function $\mathrm{Re}\tr(KX)$ attains a maximum on this set. If it equaled $\norm X$, full rank would require $K=V^\dagger$, which the main proof excludes for separable effects with phases. Thus $\CSEP<\norm X$. Since LOCC measurements form a subset of separable measurements, the strict upper bound applies to protocols with arbitrary rounds of classical communication.

\section{Residue derivation of the accelerated-detector rates}\label{app:unruh}

For the contour argument, replace $i0$ temporarily by $i\varepsilon$ with $0<\varepsilon<\pi/a$, and denote the regulated function by $W_\varepsilon$. The limit $\varepsilon\to0^+$ is taken after integration. We use the spectral convention
\begin{equation}
\mathcal W(\omega)=\int_{-\infty}^{\infty}d\xi\,e^{i\omega\xi}W(\xi),\qquad
\gamma_\downarrow=\epsilon^2\mathcal W(\Omega),\qquad
\gamma_\uparrow=\epsilon^2\mathcal W(-\Omega).
\label{eq:spectral-convention}
\end{equation}
For $\omega>0$, the relation $|e^{i\omega(x+iy)}|=e^{-\omega y}$ selects closure in the upper half-plane. The second-order poles of Eq.~\eqref{eq:unruh-W} are
\begin{equation}
\xi_n=i\varepsilon+\frac{2\pi i n}{a},\qquad n=0,1,2,\ldots.
\label{eq:unruh-poles}
\end{equation}
Near $\xi=\xi_n$, $\sinh[a(\xi-i\varepsilon)/2]=(-1)^n[a(\xi-\xi_n)/2+O((\xi-\xi_n)^3)]$, so
\begin{equation}
W_\varepsilon(\xi)=-\frac1{4\pi^2(\xi-\xi_n)^2}+O(1),\qquad
\mathop{\mathrm{Res}}_{\xi=\xi_n}\bigl[e^{i\omega\xi}W_\varepsilon(\xi)\bigr]
=-\frac{i\omega}{4\pi^2}e^{i\omega\xi_n}.
\label{eq:unruh-residue}
\end{equation}
Choose a rectangle with vertical edges at $\pm L$ and upper edge at $Y_N=\varepsilon+(2N+1)\pi/a$. On the upper edge, $|W_\varepsilon(x+iY_N)|=a^2/[16\pi^2\cosh^2(ax/2)]$. For fixed $N$, the vertical-edge contributions are $O(Y_Ne^{-aL})$ and vanish as $L\to\infty$. The remaining horizontal edge obeys
\begin{equation}
\left|\int_{-L}^{L}dx\,e^{i\omega(x+iY_N)}W_\varepsilon(x+iY_N)\right|
\le\frac{a}{4\pi^2}e^{-\omega Y_N}\longrightarrow0.
\label{eq:contour-bound}
\end{equation}
Letting $N\to\infty$ and then $\varepsilon\to0^+$, the residue sum gives
\begin{equation}
\begin{aligned}
\mathcal W(\omega)&=\lim_{\varepsilon\to0^+}2\pi i\sum_{n=0}^{\infty}
\left(-\frac{i\omega}{4\pi^2}\right)e^{-\omega\varepsilon-2\pi n\omega/a}\\
&=\frac{\omega}{2\pi(1-e^{-2\pi\omega/a})}.
\end{aligned}
\label{eq:positive-spectrum}
\end{equation}
For negative frequency $\omega=-\Omega<0$, close clockwise in the lower half-plane. Taking the lower edge at $Y_N^-=\varepsilon-(2N+1)\pi/a$ gives the same horizontal bound with $e^{\Omega Y_N^-}\to0$; the vertical-edge estimate is unchanged. The poles have $n=-m$ with $m=1,2,\ldots$ and residues $i\Omega e^{\Omega\varepsilon-2\pi m\Omega/a}/(4\pi^2)$. Thus
\begin{equation}
\begin{aligned}
\mathcal W(-\Omega)&=\lim_{\varepsilon\to0^+}\frac{\Omega e^{\Omega\varepsilon}}{2\pi}
\sum_{m=1}^{\infty}e^{-2\pi m\Omega/a}\\
&=\frac{\Omega e^{-2\pi\Omega/a}}{2\pi(1-e^{-2\pi\Omega/a})}=q\mathcal W(\Omega).
\end{aligned}
\label{eq:negative-spectrum}
\end{equation}
Multiplication by the corresponding $\epsilon_\nu^2$ gives Eq.~\eqref{eq:unruh-rates}.

\clearpage
\section*{Statements and Declarations}
Funding: No funding was received for this work.

Competing interests: The author declares no competing interests.

Author contributions: Zhiming Huang is responsible for the conception, analysis, validation, and writing of the manuscript.

Data availability: No external datasets were used or collected in this theoretical study. All quantities underlying the figure are specified by the analytical expressions in the article.

\begingroup\footnotesize\renewcommand{\em}{\normalfont}

\endgroup
\end{document}